\documentclass[10pt, conference]{IEEEtran}
\IEEEoverridecommandlockouts

\usepackage{latexsym}
\usepackage{graphicx}
\usepackage{amsfonts,amssymb,amsmath,graphicx,tikz, bbold}
\usepackage{hyperref} 
\usepackage{etoolbox}
\usepackage{array}

\newtheorem{theorem}{Theorem}
\newtheorem{lemma}[theorem]{Lemma}
\newtheorem{corollary}[theorem]{Corollary}
\newenvironment{proof}{\textit{Proof:}}{\hfill$\square$\\}

\newcounter{constraint}[section]

\renewcommand{\mod}{ \text{ mod }  }

\newcommand{\dm}[1]{\MakeUppercase{#1}}

\newcommand{\rvec}[1]{\boldsymbol{\MakeLowercase{#1}}}

\newcommand{\dvec}[1]{\MakeLowercase{#1}}

\newcommand{\tp}{\mathrm{T}}

\newcommand{\Cesaro}{Ces\'{a}ro }

\title{\LARGE \bf
Time-invariant prefix-free source coding for MIMO LQG control
}

\author{Travis C. Cuvelier, Takashi Tanaka, and Robert W. Heath, Jr.
\thanks{This work was supported in part by the National Science Foundation under Grant No. ECCS-1711702 and and CAREER Award \#1944318. }
\thanks{T. Cuvelier is with the Department
of Electrical and Computer Engineering, The University of Texas at Austin,
TX, 78712 USA e-mail: tcuvelier@utexas.edu.}
\thanks{T. Tanaka is with the Department
of Aerospace Engineering and Engineering Mechanics, The University of Texas at Austin,
TX, 78712 USA e-mail: ttanaka@utexas.edu.}
\thanks{R. Heath is with the Department
of Electrical and Computer Engineering, North Carolina State University, Raleigh, 
NC, 27606 USA e-mail: rwheath2@ncsu.edu.}}

\begin{document}

\maketitle

\begin{abstract}
In this work we consider discrete-time multiple-input multiple-output (MIMO) linear-quadratic-Gaussian (LQG) control where the feedback consists of variable length binary codewords. To simplify the decoder architecture, we enforce a strict prefix constraint on the codewords. 
We develop a data compression architecture that provably achieves a near minimum time-average expected bitrate for a fixed constraint on the LQG performance. The architecture conforms to the strict prefix constraint and does not require time-varying lossless source coding, in contrast to the prior art. 
\end{abstract}

\section{Introduction}\label{sec:introduction} In this work, we consider discrete-time MIMO LQG control where feedback occurs via variable-length packets (codewords) of bits that we assume are conveyed reliably (without error) from a sensor/encoder to a decoder/controller. For some constraint on LQG control performance, we seek to design a sensor/encoder and decoder/controller that minimize the time-averaged expected bitrate of the binary channel.

The motivation for this formulation is communication-efficient remote control over wireless communications. In particular, we imagine a scenario where a remote sensor platform measures some dynamical system and conveys the measurements to the controller over a noiseless binary channel. At the physical and medium access control layers, reliable control of autonomous agents over wireless has motivated the development of ultra low latency reliable communication (ULLRC). While progress has been made on the development ULLRC, the fundamental scarcity of physical layer resources motivates approaches that minimize the use of these resources. In contrast to the work on ULLRC, we aim higher on the protocol stack; namely we propose to develop control and data compression (source coding) strategies to achieve satisfactory control performance with minimal communication bitrate. The minimum average length, in bits, of the binary packets the sensor encodes are an effective surrogate for the minimum amount of physical layer resources (time/bandwidth/power) required to achieve satisfactory performance for a fixed reliability (cf. e.g. \cite[Section 2.6]{hjjournal}). We impose a prefix constraint on the codewords that ensures that other users sharing the same communication medium can identify the end of each transmitted message and thus identify the channel as free-to-use. In contrast to the prior art, this work imposes a \textit{time-invariant} (TI) prefix constraint on the codewords, which simplifies the decoding architecture architecture; the end of codewords can be detected via comparison with a TI list \cite{ourlblett}. 

In the prior literature, achievability results for the problem of interest follow from asymptotically bounding the output entropy of a quantizer \cite{silvaFirst}\cite{tanakaISIT}\cite{kostinaTradeoff}. While these bounds can nearly be achieved with zero-delay lossless source coding, they generally require that the lossless code be adapted, at every timestep, to the probability mass function of the quantizer output (cf. \cite[Section IV.A]{ourJSAIT}). This adds a great deal of computational complexity to the encoder and decoder.

In this work, we propose a data compression architecture for MIMO plants based on \cite{tanakaISIT} and \cite{ourJSAIT} that attains the desired LQG performance, satisfies a TI prefix constraint, and does not require a time-varying lossless source codec. Starting from the architecture of \cite{tanakaISIT}, we prove that the quantizer output has a limiting distribution. We propose to losslessly encode the quantizations with a fixed prefix-free code adapted the limiting distribution. We prove that under this modification, the codewords satisfy a TI prefix constraint without an increase in codeword length and without the complexity of time-varying lossless source coding. The architecture nearly achieves a known lower bound on the minimum expected bitrate. In particular,  we use results from systems theory to extend the scalar TI achievability approach in \cite{ourJSAIT} to the MIMO setting. 

\subsection{Related Work}
Our work follows from the architecture for LQG control with minimum bitrate prefix-free variable-length coding in \cite{silvaFirst}. For scalar systems, \cite{silvaFirst} derived a lower bound on the bitrate of a prefix-free source codec inserted into an LQG feedback channel in terms of Massey's directed information (DI) \cite{masseyDI}. This motivated a rate-distortion problem for the tradeoff between DI and LQG control performance. Using entropy-coded dithered quantization (ECDQ), \cite{silvaFirst} proved that the DI lower bound was nearly achievable. In ECDQ, uses a shared sequence of IID uniform random variables to whiten the error induced by quantization process. The quantization are then encoded via an entropy source code (e.g. a Shannon-Fano-Elias code). The rate-distortion formulation was generalized to MIMO plants in \cite{SDP_DI}. Likewise, the achievability approach was generalized in \cite{tanakaISIT}. New analytical lower bounds (for MIMO systems) for the bitrate/control performance tradeoff were derived in \cite{kostinaTradeoff}. It was further demonstrated that such bounds were nearly achievable \textit{without} the use of a dither signal in the high communication rate/strict control cost regime \cite{kostinaTradeoff}.  The achievablilty approaches in \cite{silvaFirst}, \cite{tanakaISIT}, and \cite{kostinaTradeoff} implicitly assume that quantizations are encoded using a time-varying lossless source code \cite{ourJSAIT}. More recently, \cite{milanCorrection} and \cite{ourlblett} refined the lower bound analysis in \cite{silvaFirst}, clarifying that the bound still applies even when the encoder and decoder share a dither signal. 

The use of time-varying source codes complicates  compression architectures; generally speaking, the codebooks used to encode quantizations must be updated at every timestep. Furthermore, in network settings, a source code's prefix properties may be used by receivers to detect the end of transmissions. If a TI prefix code is used, the set of prefixes against which received codewords are compared is constant over time \cite{ourlblett}. In the time-varying case, there are no such guarantees. In \cite{ourJSAIT} it was shown that when a dither signal is available, there exists a TI quantization and coding scheme for scalar plants that nearly achieves the DI lower bound. In this work, we generalize this to the MIMO setting; namely we prove the existence of a quantizer and TI source codec for MIMO plants that nearly achieves the directed information lower bound \cite{SDP_DI}. We believe this to be the first such ``time-invariant" achievability  result for MIMO plants in the literature. 

Work on fixed bitrate and event-driven communication strategies for control are also relevant to this work. In \cite{fixlenalgs}, numerical experiments demonstrated that Lloyd-Max style quantization could be used to attain a control performance competitive with known variable-rate upper bounds. However, the quantizer in \cite{fixlenalgs} is time-varying, and many experimental realizations had time-average control costs that greatly exceeded the variable-length bounds \cite{fixlenalgs}. More recently, \cite{lowbrperiodic} have proposed the the mean time to quantizer saturation as a metric for the analysis of fixed-length coding in LQG feedback control. Periodic coding strategies were devised, and their performances and  escaped times were analyzed theoretically and with an experiment \cite{lowbrperiodic}. More recently \cite{guo2021optimal} considered tracking a scalar Markov source under an event-based communication paradigm. Under this formulation, an encoder monitoring a plant can send binary codewords (without prefix constraints) at arbitrary (continuous-time) instants to a synchronized decoder \cite{guo2021optimal}. For a constraint on estimator error, the minimum bitrate communication policy was deduced. In this work, as in \cite{tanakaISIT}, \cite{kostinaTradeoff}, \cite{fixlenalgs}, we consider a sampled discrete-time control and communication model. While event-driven schemes like \cite{guo2021optimal} may be able to achieve smaller bitrates than these works, the models for control and communication in \cite{tanakaISIT}, \cite{kostinaTradeoff}, \cite{fixlenalgs} are more amenable to real-world sampled data systems with finite sensing and communication bandwidths.  
\subsection{Notation} We denote constant scalars and vectors in lowercase $x$, scalar and vector random variables in boldface $\rvec{x}$, and matrices by capital letters $X$. We let $[\dvec{x}]_{i}$ denote the $i^{\mathrm{th}}$ element of $\dvec{x}$, $I_{m}$ denote the $\mathbb{R}^{m\times m}$ identity matrix,  $0_{m\times m}$ the respective zero matrix. Let $\lVert X \rVert_{2}$ denote the largest singular value of $X$. We write P(S)D for ``symmetric positive (semi)definite", and let $\mathbb{S}^{m}_{+}$ denote the set of $m\times m$ PSD matrices. We let $\succ$, $\succeq$ denote the standard partial order on the PSD cone, e.g. if $A,B\in\mathbb{R}^{m}$, we write $A\succ B$ if $A-B$ is PD, likewise $A\succeq B$ if $A-B$ is PSD. 
For time domain sequences, let $\{\rvec{x}_{t}\}$ denote $(\rvec{x}_{0},\rvec{x}_{1},\dots)$, $\rvec{x}_{a}^{b}$ denote $(\rvec{x}_{a},\dots,\rvec{x}_{b})$ if $b\ge a$, $\rvec{x}_{a}^{b}=\emptyset$ otherwise. Let $\rvec{x}^{b}= \rvec{x}_{0}^{b}$. Denote the set of finite-length binary strings by $\{0,1\}^*$. For $a\in \{0,1\}^{*}$, let $\ell(a)$ be the length of $a$ in bits. Let $\Delta\mathbb{Z}^{m}$ denote the set of $m-$tuples whose elements are integer multiples of $\Delta$. Define the set ${\mathcal{B}}^{m}(\Delta)= \{\dvec{x}\in\mathbb{R}^{m}: \lVert\dvec{x} \rVert_{\infty}\le \frac{\Delta}{2} \}$. For a set $\mathcal{S}$, define the indicator function of $\dvec{x}\in\mathcal{S}$ as $1_{\dvec{x}\in\mathcal{S}}$.  For a topological space $\mathbb{T}$, let $\mathbb{B}(\mathbb{T})$ denote the standard Borel $\sigma$-algebra of $\mathbb{T}$. For Euclidean spaces, let $\lambda$ denote the Lebesgue measure. We use PMF for ``probability mass function", PDF for ``probability density function" (with respect to $\lambda$),  $\rvec{a}\perp\rvec{b}$ for ``$\rvec{a}$ and $\rvec{b}$ are independent", and $H$ for entropy. 

\section{System Model and Problem Formulation}\label{sec:systemmodel}
We consider the system model depicted in Fig. \ref{fig:ditharch}. The system to be controlled is a TI, multidimensional, linear dynamical system  (i.e., a MIMO plant) plant controlled via a feedback model where communication occurs over an ideal (delay and error free) binary channel. The plant is fully observable to an encoder/sensor block, which conveys a variable-length binary codeword $\rvec{a}_{t}\in\{0,1\}^*$ over the channel to a combined decoder/controller. Upon receipt of the codeword, the decoder/controller designs the control input. Denote the state vector $\rvec{x}_t\in \mathbb{R}^{m}$, the control input $\rvec{u}_{t}\in\mathbb{R}^{u}$, and let $\rvec{w}_{t}\sim\mathcal{N}(\dvec{0},{W})$ denote processes noise assumed to be IID over time. We assume ${W}\succ{0}_{m\times m}$, i.e., the process noise covariance is full rank. We assume assume that $\rvec{x}_{0}\sim\mathcal{N}(\dvec{0},{X}_0)$ for some $X_0\succeq 0$.  For ${A}\in\mathbb{R}^{m\times m}$ the system matrix  and ${B}\in\mathbb{R}^{m\times u}$ the  feedback gain matrix, for $t\ge 0$ the plant dynamics are given by $\rvec{x}_{t+1} = {A}\rvec{x}_{t}+{B}\rvec{u}_{t}+\rvec{w}_{t}$.  To ensure finite control cost is attainable, we assume $({A},{B})$ are stabilizable. We assume that the encoder and decoder share access to a random uniform dither signal $\{\rvec{d}_{t}\}$. We assume that, for some $\Delta>0$ to be specified, the random vectors $\rvec{d}_{t}\in\mathbb{R}^m$ have components that are independently uniformly distributed on $[-\Delta/2,\Delta/2]$, and that the sequence $\{\rvec{d}_{t}\}$ is IID over time.  We assume that  $\{\rvec{w}_{t}\}$, $\{\rvec{d}_{t}\}$, and $\rvec{x}_{0}$ are mutually independent. In real-world systems, this \textit{shared randomness} can be effectively accomplished using synchronized pseudorandom number generators. We assume that the encoder/sensor and the decoder/controller may be randomized given their inputs. The encoder/sensor policy in Fig. \ref{fig:ditharch} is a sequence of causally conditioned Borel measurable kernels denoted $\mathbb{P}_{\mathrm{E}}[\rvec{a}_{0}^{\infty}|| \rvec{d}_{0}^{\infty},\rvec{x}_{0}^{\infty}] = \left\{ \mathbb{P}_{\mathrm{E}}[\rvec{a}_{t}|\rvec{a}_{0}^{t-1},\rvec{d}_{0}^{t},\rvec{x}_{0}^{t}]\right\}_t$, and that 
 corresponding decoder/controller policy is given by $\mathbb{P}_{\mathrm{C}}[\rvec{u}^{\infty}|| \rvec{a}^{\infty},\rvec{d}^{\infty}] = \left\{ \mathbb{P}_{\mathrm{C}}[\rvec{u}_{t}|\rvec{a}^{t},\rvec{d}^{t},\rvec{u}^{t-1}]\right\}_t$. 
 Note that under the assumed dynamics, $\rvec{x}^{t}$ is a deterministic function of $\rvec{x}$, $\rvec{a}^{t-1}$, $\rvec{u}^{t-1}$, and $\rvec{w}^{t-1}$.  We encode  conditional independence assumptions in the system model by a factorization of the one-step transition kernels for $\rvec{a}_{t}$, $\rvec{d}_{t}$, $\rvec{u}_{t}$, and  $\rvec{w}_{t}$. For $t\ge 0$, we assume the kernels factorize via
 \begin{subequations}\label{eq:ditherFactorization}
\begin{multline}
      \mathbb{P}[\rvec{a}_{t+1},\rvec{u}_{t+1}|\rvec{a}^{t},\rvec{d}^{t+1},\rvec{u}^{t},\rvec{w}^{t},\rvec{x}_{0}] =\\\mathbb{P}_{\mathrm{E}}[\rvec{a}_{t+1}|\rvec{a}^{t},\rvec{d}^{t+1},\rvec{x}^{t+1}]\mathbb{P}_{\mathrm{C}}[\rvec{u}_{t+1}|\rvec{a}^{t+1},\rvec{d}^{t+1},\rvec{u}^{t}],
\end{multline}   
\begin{multline} \mathbb{P}[\rvec{a}_{t+1},\rvec{d}_{t+1},\rvec{u}_{t+1},\rvec{w}_{t+1}|\rvec{a}^{t},\rvec{d}^{t},\rvec{u}^{t},\rvec{w}^{t},\rvec{x}_{0}] = \\\mathbb{P}[\rvec{a}_{t+1},\rvec{u}_{t+1}|\rvec{a}^{t},\rvec{d}^{t},\rvec{u}^{t},\rvec{w}^{t},\rvec{x}_{0}]\mathbb{P}[\rvec{d}_{t+1}]\mathbb{P}[\rvec{w}_{t+1}],
\end{multline}
\end{subequations} and that we have initially $\mathbb{P}[\rvec{a}_0,\rvec{d}_0,\rvec{u},\rvec{w}_{0}|\rvec{x}_{0}]$ $=$ $\mathbb{P}[\rvec{w}_{0}]\mathbb{P}[\rvec{d}_{0}]\mathbb{P}_{\mathrm{E}}[\rvec{a}_{0}|\rvec{x}_{0},\rvec{d}_{0}]$ $\mathbb{P}_{\mathrm{C}}[\rvec{u}_{0}|\rvec{a}_{0},\rvec{d}_0]$. Implications of these factorization are discussed in the caption of Fig. \ref{fig:ditharch}. 
\begin{figure}
	\centering
	\includegraphics[scale = .2]{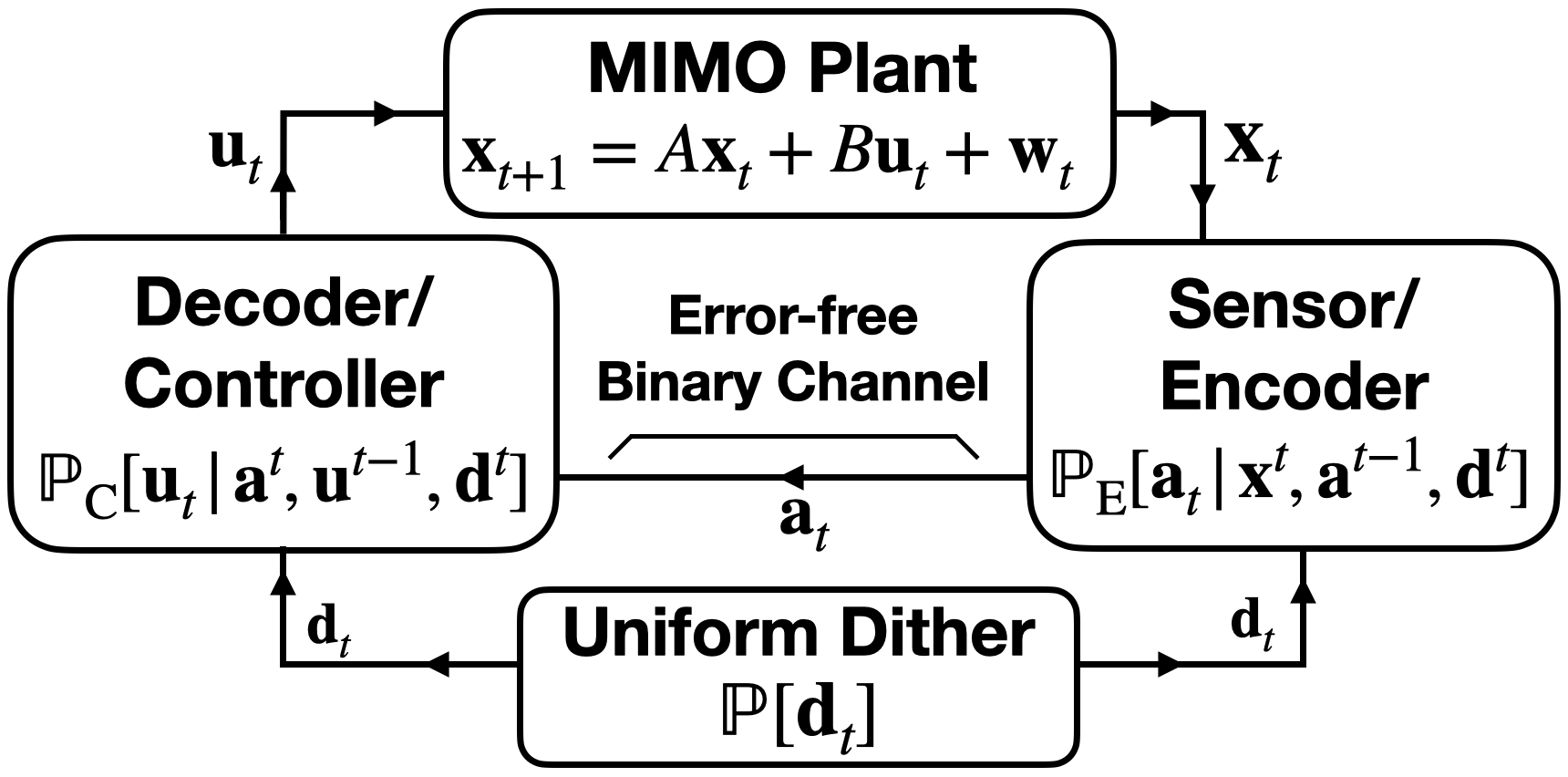}
	\vspace{-.2cm}
	\caption{ The encoder can select the codeword $\rvec{a}_{t}$ randomly given ``the information it knows at time $t$". When $\rvec{a}_{t}$ arrives at the decoder, the decoder can randomly generate the control $\rvec{u}_{t}$ given $\rvec{a}_{t}$ as well as its own prior knowledge. The encoder and decoder share access to $\rvec{d}_{t}$, which is IID and generated ``independently" of all past system variables.  }\label{fig:ditharch}
\end{figure}

We require the codewords to conform to a TI prefix constraint, namely that for all $i,j$ and distinct $a_1,a_2\in\{0,1\}^*$ with both $\mathbb{P}[\rvec{a}_i=a_1]>0$ and $\mathbb{P}[\rvec{a}_j=a_2]>0$, $a_1$ is not a prefix of $a_2$ and vice-versa. Thus, not only must the support of the random variable $\rvec{a}_{i}$ be a set of prefix-free binary codewords, we also require the union of the supports of the $\{\rvec{a}_{i}\}$ be prefix-free. This is a stronger prefix constraint than was considered in \cite{tanakaISIT}. The TI constraint may enable for more computationally efficient communication resource sharing; the end of a codeword can be unambiguously identified via comparing received messages against the TI union of supports. We are interested in the following optimization, over policies $\mathbb{P}_\mathrm{E}, \mathbb{P}_{\mathrm{C}}$ that conform to the prefix constraint:
\begin{equation}\label{eq:codewordLenghtOptimization}
\mathcal{L}(\gamma) = \left\{ \begin{aligned}
& \underset{\mathbb{P}_\mathrm{E}, \mathbb{P}_{\mathrm{C}}}{\inf}  \text{ }\frac{1}{T}\sum\nolimits_{t=0}^{T-1}\mathbb{E}[\ell(\rvec{a}_{t})] \\ &\text{s.t. }   \frac{1}{T}\sum\nolimits_{t=0}^{T-1}\mathbb{E}[\lVert \rvec{x}_{t+1} \rVert_{{Q}}^{2} +\lVert \rvec{u}_{t} \rVert_{{\Phi}}^{2}] \le \gamma,
\end{aligned} \right.
\end{equation} where ${Q}\succeq {0}$, ${\Phi}\succ {0}$, and $\gamma$ is the maximum tolerable LQG cost. Let $S$ be a stabilizing solution to the discrete algebraic Riccati equation (DARE) $A^{\mathrm{T}}SA-S-A^{\mathrm{T}}SB(B^{\mathrm{T}}SB+\Phi)^{-1}B^{\mathrm{T}}SA+Q = 0$, $K=-(B^{\mathrm{T}}SB+\Phi)^{-1}B^{\mathrm{T}}SA$, and $\Theta = K^{\mathrm{T}}(B^{\mathrm{T}}SB+\Phi)K$. Consider the optimization 
\begin{align}\label{eq:threestageRDF}
    \mathcal{R}(\gamma) &=& \left\{ \begin{aligned}
& \underset{\substack{P,\Pi, \in\mathbb{R}^{m\times m}\\P,\Pi\succeq 0} }{\inf} \frac{1}{2}(-\log_{2}{\det{\Pi}}+\log_{2}{\det{W}} )\\ &\text{ }\text{s.t. }  \mathrm{Tr}(\Theta P)+\mathrm{Tr}(WS)\le \gamma\text{,  }\\&\text{ } P\preceq APA^\mathrm{T}+W\text{, } \\&\text{ }\text{ }\text{ }\begin{bmatrix}P-\Pi & PA^{\mathrm{T}} \\ AP & APA^{\mathrm{T}}+W \end{bmatrix}\succeq0 
\end{aligned}\right. 
\end{align} A consequence of \cite{ourlblett} (cf. \cite{ourJSAIT} and \cite{SDP_DI}) is the lower bound 
$\mathcal{R}(\gamma)\le \mathcal{L}(\gamma)$. In the sequel, we will analyze the bitrate of the proposed achievability scheme in terms of $\mathcal{R}(\gamma)$.

\section{Achievability Architecture} 
\begin{table}[]
\begin{tabular}{ccllc}
\hline
\multicolumn{1}{|c|}{Variable}                & \multicolumn{4}{c|}{Description/Key Equations}      \\ \hline\hline
\multicolumn{1}{|c|}{$\hat{P}\in\mathbb{S}_{+}^{m}$}                     & \multicolumn{4}{c|}{Minimizing $P$ of (\ref{eq:threestageRDF}), asymptotic KF estimator error}                                                  \\ \hline
\multicolumn{1}{|c|}{$\hat{P}_{+}\in\mathbb{S}_{+}^{m}$}                     & \multicolumn{4}{c|}{$\hat{P}_{+}=A\hat{P}A^{\mathrm{T}}+W$, asymptotic KF prediction error }                                                  \\ \hline
\multicolumn{1}{|c|}{$C\in\mathbb{R}^{m\times m}$}                     & \multicolumn{4}{c|}{Measurement matrix, $\hat{P}^{-1}-\hat{P}_{+}^{-1}=C^{\mathrm{T}}C\frac{12}{\Delta^{2}}$ }                                                  \\ \hline
\multicolumn{1}{|c|}{$\Delta\in \mathbb{R}$} & \multicolumn{4}{c|}{Quantizer/dither sensitivity, $\hat{P}^{-1}-\hat{P}_{+}^{-1}=C^{\mathrm{T}}C\frac{12}{\Delta^{2}}$}                                 \\ \hline
\multicolumn{1}{|c|}{$K\in\mathbb{R}^{u\times m}$}                     & \multicolumn{4}{c|}{Certainty equivalent feedback control gain}                                \\ \hline
\multicolumn{1}{|c|}{$\overline{\rvec{x}}_{t|t-1} \in\mathbb{R}^{m}$}                     & \multicolumn{4}{c|}{TI KF prediction of $\mathbf{x}_{t}$}                                \\ \hline 
\multicolumn{1}{|c|}{$\overline{\rvec{x}}_{t|t} \in\mathbb{R}^{m}$}                     & \multicolumn{4}{c|}{TI KF estimate of $\mathbf{x}_{t}$. $\mathbf{u}_{t}=K\overline{\rvec{x}}_{t|t}$}                                \\ \hline 
\multicolumn{1}{|c|}{$\overline{\rvec{e}}_{t} \in\mathbb{R}^{m}$}                     & \multicolumn{4}{c|}{KF predict error $\overline{\rvec{e}}_{t}=\overline{\rvec{x}}_{t|t-1}-\mathbf{x}_{t}$, $\mathbb{E}[\overline{\rvec{e}}_{t}\overline{\rvec{e}}_{t}^{\mathrm{T}}]\rightarrow \hat{P}_{+}$}                                \\ \hline \multicolumn{1}{|c|}{$\rvec{d}_{t} \in\mathbb{R}^{m}$}                     & \multicolumn{4}{c|}{Dither sequence, IID elements $\sim\text{Uniform}(\Delta)$}                                \\ \hline
\multicolumn{1}{|c|}{$\rvec{q}_{t} \in\Delta\mathbb{Z}^{m}$}                     & \multicolumn{4}{c|}{Quantization $\rvec{q}_{t}=Q_{\Delta}(C\rvec{e}_{t}+\rvec{d}_{t})$, encoded into $\rvec{a}_{t}$}                                \\ \hline \multicolumn{1}{|c|}{$\rvec{\tilde{q}}_{t} \in\mathbb{R}^{m}$}                     & \multicolumn{4}{c|}{Reconstruction $\rvec{\tilde{q}}_{t}=\rvec{q}_{t}-\rvec{d}_{t}$,  error $\rvec{v}_{t}=\rvec{\tilde{q}}_{t}-C\mathbf{e}_{t}$}      
      \\ \hline \multicolumn{1}{|c|}{$\rvec{v}_{t} \in\mathbb{R}^{m}$}                     & \multicolumn{4}{c|}{$\rvec{v}_{t}\perp\mathbf{e}_{t}$, $\rvec{v}_{t}\perp\mathbf{x}_{t}$, $\mathbb{E}[\rvec{v}_{t}\rvec{v}_{t+k}^{\mathrm{T}}] = I_{m}\frac{\Delta^2}{12}\mathbb{1}_{k=0}$ }   
\\ \hline \multicolumn{1}{|c|}{$\rvec{y}_{t} \in\mathbb{R}^{m}$}                     & \multicolumn{4}{c|}{Effective measurement at decoder $\rvec{y}_{t}=C\mathbf{x}_{t}+\mathbf{v}_{t}$}                                \\ \hline\\
\end{tabular}
\caption{}
\vspace{-.9cm}
\end{table}
The achievability approach we proposed is demonstrated in Fig. \ref{fig:achievability}. It is almost identically to the time-varying MIMO achievability approach in \cite[Section IV.B]{ourJSAIT} with the exception that the lossless Shannon-Fano-Elias source codec used to encode the quantizations is TI. We first describe the signals in Fig. \ref{fig:achievability} before summarizing relevant results from \cite{ourJSAIT} that simplify the analysis. Let $\hat{P}$ denote the minimizing $P$ from (\ref{eq:threestageRDF}), let $\hat{P}_{+}= A\hat{P}A^{\mathrm{T}}+W$, 
and let $\Delta \in \mathbb{R}$, $\Delta>0$, and $C\in \mathbb{R}^{m\times m}$ be such that     $\hat{P}^{-1}-\hat{P}_{+}^{-1}=C^{\mathrm{T}}C\frac{12}{\Delta^{2}}$. In this way, $\hat{P}$ ($\hat{P}_{+}$ ) is interpreted as the asymptotic (prediction) error covariance attained by a Kalman filter tracking the source process $\{\rvec{x}_{t}\}$ under a measurement model $\rvec{y}_{t}=C\rvec{x}_{t}+\rvec{v}_{t}$, where $\mathbb{E}[\rvec{v}_{t}\rvec{v}_{t+k}^{\mathrm{T}}] = I_{m}\frac{\Delta^2}{12}\mathbb{1}_{k=0}$, $\mathbb{E}[\rvec{v}_{t}\rvec{x}_{0}^{\mathrm{T}}]=\mathbb{E}[\rvec{v}_{t}\rvec{w}_{t+k}^{\mathrm{T}}]=0_{m\times m}$. We will show that this measurement model is attained by the decoder in Fig. \ref{fig:achievability} using dithered quantization. 
\begin{figure}[h] 
	\centering
	\includegraphics[scale = .231]{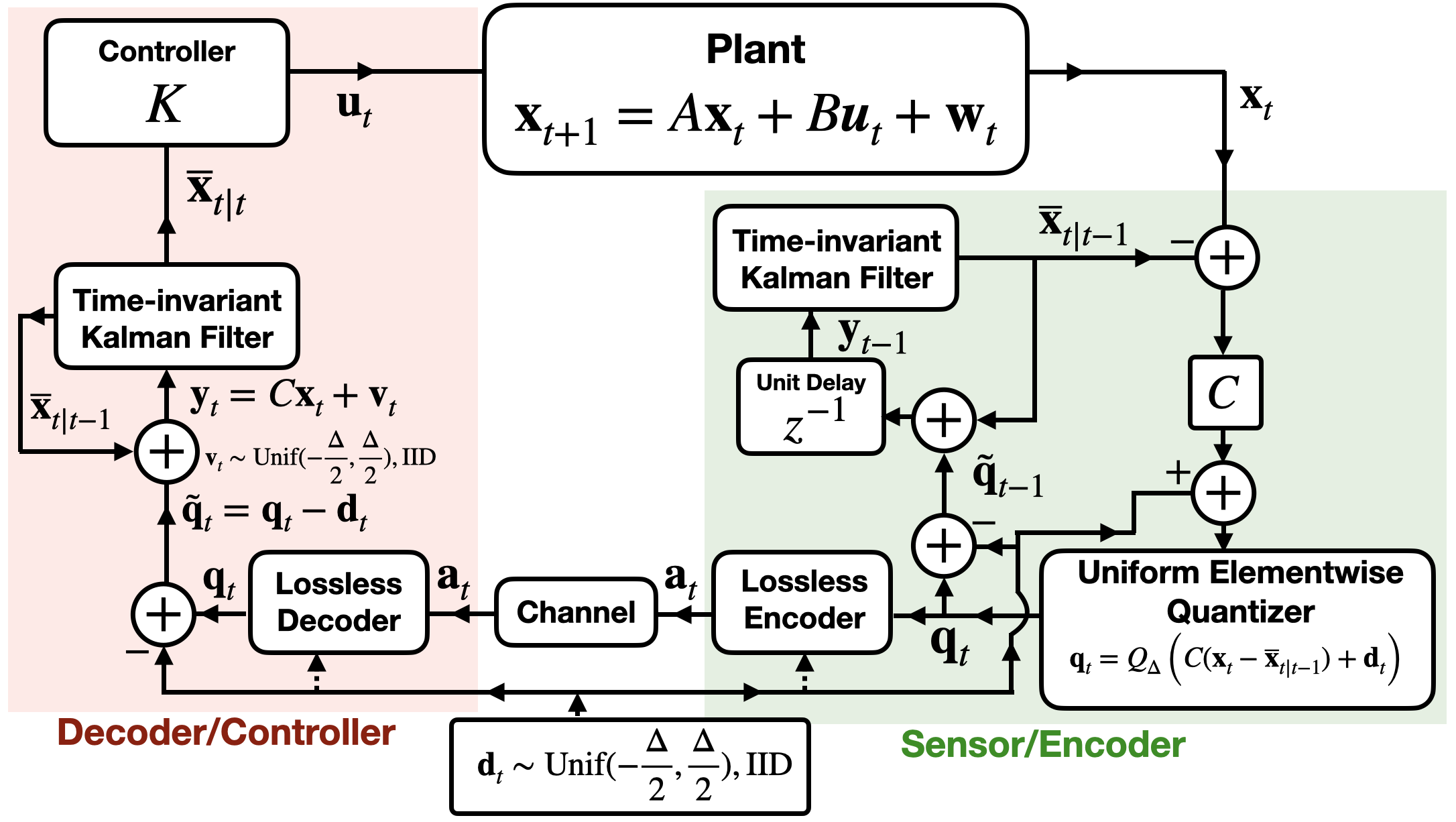}
    \vspace{-.3cm}
	\caption{The achievability architecture.  Descriptions of the variables may be found in Table I. }    \vspace{-.26cm}
\label{fig:achievability}
\end{figure}

Define a uniform elementwise quantizer via
$Q_{\Delta}\colon\mathbb{R}^{m}\rightarrow\Delta\mathbb{Z}^{m}$ via $  [Q_{\Delta}(\dvec{x})]_{i} =  k\Delta, \text{ if }[\dvec{x}]_{i}\in[k\Delta-\Delta/2,k\Delta+\Delta/2),$
 e.g. each element of the vector $\dvec{x}$ is rounded to the nearest multiple of $\Delta$.  Both the encoder and decoder in Fig. \ref{fig:achievability} operate identical time invariant Kalman filters that compute recursive estimates of $\rvec{x}_{t}$ given measurements received by the decoder. Let $\overline{\rvec{x}}_{t|t-1}$ denote the prior estimate at time $t$ and $\rvec{\overline{x}}_{t|t}$ denote the posterior. Set $\rvec{\overline{x}}_{0|-1}=0$. At every time $t$ the encoder and decoder produces the Kalman filter prediction error 
\begin{align}\label{eq:kferrordef}
    \rvec{e}_{t} = \rvec{x}_{t}-\overline{\rvec{x}}_{t|t-1},
\end{align} and the linear measurement innovation $ \rvec{e}_{t}$. It then produces the discrete, dithered quantization $q_{t} = Q_{\Delta}(C\rvec{e}_{t}+\rvec{d}_{t})$. It encodes the discrete quantization losslessly into the codeword $\rvec{a}_{t}$ which it transmits to the decoder. Upon receiving $\rvec{a}_{t}$, the decoder can recover $\rvec{q}_{t}$ exactly. It then produces the reconstruction $\rvec{\tilde{q}}_{t} =\rvec{q}_{t}-\rvec{d}_{t}$. Let $\rvec{v}_{t} = \rvec{\tilde{q}}_{t}-C\rvec{e}_{t}$ denote the reconstruction error. It then produces the centered measurement $\rvec{y}_{t}=\rvec{\tilde{q}}_{t}+C\rvec{x}_{t|t-1}$ (equivalently $\rvec{y}_{t} =C\rvec{x}_{t}+\rvec{v}_{t}$). By \cite[Lemma 1]{tanakaISIT}, the elements  $[\rvec{v}_{t}]_{i}$ are mutually independent and uniform on $[-\Delta/2,\Delta/2]$ and furthermore $\rvec{v}_{t}\perp \rvec{x}_{t}$. Denote the TI Kalman filter gain $J= \hat{P}_{+}{C}^{\mathrm{T}}({C}\hat{P}_{+}{C}^{\mathrm{T}}+I_{m\times m}\frac{\Delta^{2}}{12})^{-1}$. The decoder updates its estimate via
\begin{align}\label{eq:askfup}
    \rvec{\overline{x}}_{t|t} = \rvec{\overline{x}}_{t|t-1}+J(\rvec{y}_{t}-C\rvec{\overline{x}}_{t|t-1}), 
\end{align} and applies the certainty equivalent control input $\rvec{u}_{t} = K \rvec{\overline{x}}_{t|t}$. It then computes the predict estimate $\rvec{\overline{x}}_{t|t-1} = {A}\rvec{\overline{x}}_{t-1|t-1}+{B}\rvec{u}_{t-1}$. Since the encoder knows the quantizations $\{\rvec{q}_{t}\}$ and the dither $\{\rvec{d}_{t}\}$, it can recover the sequence of measurements $\{\rvec{y}_{t}\}$. The encoder can thus compute the sequence of TI Kalman filter estimates $\{\rvec{\overline{x}}_{t|t-1},\rvec{\overline{x}}_{t|t}\}$. While we have yet to define how $\rvec{q}_{t}$ is encoded into $\rvec{a}_{t}$, we have the following. 
\begin{lemma}\label{lemm:summarylemma}
So long as the decoder recovers $\rvec{q}_{t}$ exactly at every $t$, we have that 1) $\rvec{v}_{t}\perp \rvec{e}^{t},\rvec{v}^{t-1} \rvec{w}^{t},\rvec{x}_{0}$, 2) $\lim_{t\rightarrow\infty}H(\rvec{q})_{t}\le \mathcal{R}(\gamma) + m\left(1+\frac{1}{2}\log_{2}\left(\frac{2\pi e}{12}\right)\right)$, 3) $\lim_{T\rightarrow \infty} \sum_{i=0}^{T} \mathbb{E}[\lVert\rvec{x}_{t} \rVert_2^{2}+\lVert\rvec{u}_{t} \rVert_2^{2} ] \le \gamma$.
\end{lemma} The first statement follows from \cite[Lemma 1]{tanakaISIT} and the system model and the subsequent are via \cite[Theorem IV.5]{ourJSAIT}. 

 The TI prefix constraint can be satisfied by encoding each $\rvec{q}_{t}$ with the same lossless prefix code at all time. Let $\rvec{q}$ be a discrete random variable on $\Delta\mathbb{Z}^{m}$ (the same alphabet as $\rvec{q}_{t}$). Assume the support of $\rvec{q}$ is such that $\mathbb{P}[\rvec{q}=k]=0$ only if $\mathbb{P}[\rvec{q}_{t}=k]=0$. Define $\overline{F}_{\rvec{q}}(q) = \mathbb{P}[\rvec{q}<q]+\frac{\mathbb{P}[\rvec{q}=q]}{2}$, and define the encoding function $C^{\mathrm{F}}_{\rvec{q}}(q) :\Delta\mathbb{Z}^{m}\rightarrow \{0,1\}^{*}$ such that $    C^{\mathrm{F}}_{\rvec{q}}(q)$ is the binary expansion of $\overline{F}_{\rvec{q}}(q)$ truncated to $\lceil-\log_{2}(\mathbb{P}_{\rvec{q}}[q])\rceil+1$ bits. This is a standard, lossless, prefix-free Shannon-Fano-Elias code for the random variable $\rvec{q}$ \cite{elemIT}. If we use $C^{\mathrm{F}}_{\rvec{q}}$ to encode $\rvec{q}_{t}$  (e.g.   $\rvec{a}_{t}=C^{\mathrm{F}}_{\rvec{q}}(\rvec{q}_{t})$) the codeword length satisfies
\begin{align}\label{eq:instantcwl}
    \mathbb{E}[\ell(\rvec{a}_{t})]\le  2+ H(\rvec{q}_{t})+D_{\mathrm{KL}}(\rvec{q}_{t}||\rvec{q}). 
\end{align} Note that if for some $k$, $\mathbb{P}[\rvec{q}=k]=0$ while $\mathbb{P}[\rvec{q}_{t}=k]>0$, then $D_{\mathrm{KL}}(\rvec{q}_{t}||\rvec{q})=\infty$. The main result of this paper is the following lemma, proved in Section \ref{sec:ergopfs}. 
\begin{lemma}\label{lemm:summaryergo}
There exists a random variable $\rvec{q}$ such that $D_{\mathrm{KL}}(\rvec{q}_{t}||\rvec{q})$ is finite for all $t$ and $\lim_{t\rightarrow \infty}D_{\mathrm{KL}}(\rvec{q}_{t}||\rvec{q})=0$.
\end{lemma} This lemma allows us to bound the bitrate for when a TI prefix-free code is used in Fig. \ref{fig:achievability}. 
\begin{theorem}
Let $\eta =  \left(1+\frac{1}{2}\log_{2}\left(\frac{2\pi e}{12}\right)\right)$. Setting $\rvec{a}_{t}=C^{\mathrm{F}}_{\rvec{q}}(\rvec{q}_{t})$ for all $t$ ensures that the system in Fig. \ref{fig:achievability} attains a control cost less than $\gamma$, satisfies the prefix constraint in Section \ref{sec:systemmodel}, and attains a communication cost that satisfies
\begin{IEEEeqnarray}{rCl}
    \lim\sup_{T\rightarrow\infty}\frac{1}{T}\sum\limits_{i=0}^{T-1}\mathbb{E}[\ell(\rvec{a}_{t})]&\le&   \mathcal{R}(\gamma)+2+m\eta \label{eq:applycesaero}
\end{IEEEeqnarray}\end{theorem}
\begin{proof}
$C^{\mathrm{F}}_{\rvec{q}}$ is a prefix free code fixed over all $t$. Furthermore, since $D_{\mathrm{KL}}(\rvec{q}_{t}||\rvec{q})$ is finite for all $t$, it is lossless on the support of \{$\rvec{q}_{t}\}$. The bound on codeword length (\ref{eq:applycesaero}) follows via applying  (\ref{eq:instantcwl}) to each term on the left-hand side and then using the \Cesaro mean (cf. Lemma \ref{lemm:summarylemma}). 
\end{proof} In Section \ref{sec:ergopfs}, we use ergodic theory to prove Lemma \ref{lemm:summaryergo}. 
\section{Proof of Lemma \ref{lemm:summaryergo}}\label{sec:ergopfs} 
We proceed as in \cite{ourJSAIT}  by analyzing the long-term behavior of the Kalman innovation and dither signals, e.g. $(\rvec{e}_{t},\rvec{d}_{t})$. At every $t$, the quantization error $\rvec{v}_{t}$ is a measurable function of $(\rvec{e}_{t},\rvec{d}_{t})$ given by $\rvec{v}_{t} = \rvec{q}_{t}-\dm{C}\rvec{e}_{t}$, or, equivalently, $\rvec{v}_{t} = Q_{\Delta}(\dm{C}\rvec{e}_{t}+\rvec{d}_{t})-\rvec{d}_{t}-\dm{C}\rvec{e}_{t}$. Let $L=AJ$ and $\dm{R}=\dm{A}-\dm{L}\dm{C}$. The following recursion for the $\rvec{e}_{t}$ (cf. (\ref{eq:kferrordef})) can be derived
\begin{IEEEeqnarray}{rCl}\label{eq:recursion}
    \rvec{e}_{t+1} &=& R \rvec{e}_{t} -L\rvec{v}_{t}+\rvec{w}_{t}\\ &=& M( \rvec{e}_{t},\rvec{d}_{t})+\rvec{w}_{t}. 
\end{IEEEeqnarray}
By construction of $L$, 
$(A-LC)$ is stable, e.g. its eigenvalues lie strictly within the complex unit circle \cite{RDE_convergence}. Define the state space $\mathbb{D}^{m} = \mathbb{R}^{m}\times \mathcal{B}^{m}(\Delta)$. For $\dvec{r},\dvec{\mu}\in\mathbb{R}^{m}$, $\Sigma \in\mathbb{S}_{+}^{m\times m}$, denote the PDF of a multivariate Gaussian random variable with mean $\dvec{\mu}$ and covariance $\Sigma$  evaluated at $\dvec{r}$ by $N(\dvec{r};\dvec{\mu},\Sigma)$, i.e. define the function  $N(\dvec{r};\dvec{\mu},\Sigma):\mathbb{R}^{m}\times \mathbb{R}^{m}\times  \mathbb{S}_{+}^{m}\rightarrow \mathbb{R}$ via 
 $N(\dvec{r};\dvec{\mu},\Sigma)=\frac{1}{\sqrt{(2\pi)^{m}\det{\Sigma}}}e^{-\frac{1}{2}(\dvec{r}-\dvec{\mu})^{\tp}\Sigma^{-1}(\dvec{r}-\dvec{\mu})}$. Define 
$M\colon (x,y)\in\mathbb{D}^{m}\rightarrow\mathbb{R}^{m}$ via
    $M(\dvec{x},\dvec{y}) = R\dvec{x}-L(Q_{\Delta}(C\dvec{x}+\dvec{y})-\dvec{y}-C\dvec{x})$. By definition, the marginal PDF of the dither is $f_{\rvec{d}_{t+1}}(\dvec{d}) = \frac{1}{\Delta^m}1_{\dvec{d}\in \mathcal{B}^{m}(\Delta)}$ for all $t$. Since $\rvec{w}_{t}\perp (\rvec{e}^{t},\rvec{d}^{t})$, and $\rvec{d}_{t+1}\perp (\rvec{e}^{t+1}, \rvec{w}^{t}$),  via (\ref{eq:recursion}) we have that $\{\rvec{e}_{t},\rvec{d}_{t}\}$ is a time-homogeneous first order Markov chain. 
    
    The PDFs of $\{\rvec{e}_{t},\rvec{d}_{t}\}$ factorize via $ f_{\rvec{e}_{t+1},\rvec{d}_{t+1}|\rvec{e}^{t},\rvec{d}^{t}}=     f_{\rvec{d}_{t+1}}f_{\rvec{e}_{t+1}|\rvec{e}_{t},\rvec{d}_{t}}$.
Via (\ref{eq:recursion}),
\begin{IEEEeqnarray}{rCl}
f_{\rvec{e}_{t+1}|\rvec{e}_{t},\rvec{d}_{t}}(\dvec{e}|\dvec{e}_{p},\dvec{d}_{p})  &=& N(\dvec{e};M(\dvec{e}_{p},\dvec{d}_{p}),W)\label{eq:normalonestep}\\f_{\rvec{e}_{t+1}|\rvec{e}_{t},\rvec{v}_{t}}(\dvec{e}|\dvec{e}_{p},\dvec{v}_{p})  &=& N(\dvec{e};R\dvec{e}_{p}-L\dvec{v}_{p},W).\label{eq:normalonestep2}
\end{IEEEeqnarray} 
Let $f_{t+1|t} = f_{\rvec{e}_{t+1},\rvec{d}_{t+1}|\rvec{e}_{t},\rvec{d}_{t}}$, and likewise $f_{t+n|t} = f_{\rvec{e}_{t+n},\rvec{d}_{t+n}|\rvec{e}_{t},\rvec{d}_{t}}$. From the foregoing, we have
\begin{align}\label{eq:transisionkernel}
    f_{t+1|t}(\dvec{e},\dvec{d}|\dvec{e}_{p},\dvec{d}_{p}) = \frac{1_{\dvec{d}\in \mathcal{B}^{m}(\Delta) }N(\dvec{e};M(\dvec{e}_{p},\dvec{d}_{p}),W)}{\Delta^m}
\end{align}
Applying the Chapman-Komogorov equations to (\ref{eq:transisionkernel}), it can be seen that the $n$-step transition kernels satisfy
\begin{align}\label{eq:nsteptrans}
f_{t+n|t}(\dvec{e},\dvec{d}|\dvec{e}_{p},\dvec{d}_{p}) =f_{\rvec{e}_{t+n}|\rvec{e}_{t},\rvec{d}_{t}}(\dvec{e}|\dvec{e}_{p},\dvec{d}_{p})\frac{1_{\dvec{d}\in \mathcal{B}^{m}(\Delta) }}{\Delta^m}. 
\end{align} We now state the generalization of \cite[Lemma A.3]{ourJSAIT}.
\begin{lemma}\label{lemm:normalPDF}
Let $\Sigma_n = \sum\limits_{i=0}^{n-1}R^{i}W (R^{\tp})^i $ and 
\begin{align}\nonumber
    \dvec{\mu}_n(\dvec{e}_0,\dvec{d}_0,\dvec{v}_1^{n-1})= R^{n-1}M(\dvec{e}_0,\dvec{d}_0)-\sum\limits_{i=0}^{n-2}R^iL\dvec{v}_{n-1-i}.
\end{align}  
For all $n\ge 1$ we have
\begin{multline}\label{eq:normalStepFormula}
    f_{\rvec{e}_{n}|\rvec{e}_{0},\rvec{d}_{0},\rvec{v}_1^{n-1}}(\dvec{e}_{n}|\dvec{e}_{0},\dvec{d}_{0},\dvec{v}_1^{n-1}) =\\ N(\dvec{e}_{n}; \dvec{\mu}_n(\dvec{e}_{0},\dvec{d}_{0},\dvec{v}_1^{n-1}), \Sigma_{n}).
\end{multline} 
\end{lemma}
\begin{proof}
The proof follows by induction on $n$. By (\ref{eq:normalonestep}), (\ref{eq:normalStepFormula}) holds for $n=1$. Assume that the relation (\ref{eq:normalStepFormula}) holds for some $n=k-1$. To avoid clutter, when obvious we suppress the arguments to PMFs, writing, e.g. $f_{\rvec{e}_{k}|\rvec{e}_0,\rvec{d}_0,\rvec{v}_{1}^{k-1}}$ in place of $f_{\rvec{e}_{k}|\rvec{e}_0,\rvec{d}_0,\rvec{v}_{1}^{k-1}}(\dvec{e}_{k}|\dvec{e}_0,\dvec{d}_0,\dvec{v}_{1}^{k-1})$. Via Bayes' Theorem, $f_{\rvec{e}_{k}|\rvec{e}_0,\rvec{d}_0,\rvec{v}_{1}^{k-1}} = \frac{f_{\rvec{e}_{k},\rvec{v}_{k-1}|\rvec{e}_0,\rvec{d}_0,\rvec{v}_{1}^{k-2}}}{f_{\rvec{v}_{k-1}|\rvec{e}_0,\rvec{d}_0,\rvec{v}_{1}^{k-2}}}$.
Since $\rvec{v}_{k-1}$ is a measurable function of  $\rvec{e}_{k-1}$ and  $\rvec{d}_{k-1}$, $\rvec{v}_{k-1}$  is conditionally independent of $(\rvec{v}_{1}^{k-2}, \rvec{e}_{0}, \rvec{d}_0 )$ given $\rvec{e}_{k-1}$. By Lemma \ref{lemm:summarylemma},  $\rvec{v}_{k-1}$ is (pairwise) independent of $\rvec{e}_{k-1}$. This further implies that $\rvec{v}_{k-1}$ is independent of $(\rvec{v}_{1}^{k-2}, \rvec{e}_{0}, \rvec{d}_0 )$. Thus we have 
\begin{IEEEeqnarray}{rCl}\label{eq:firstSubErgo}
    f_{\rvec{v}_{k-1}} &=& f_{\rvec{v}_{k-1}|\rvec{e}_{k-1},\rvec{e}_0,\rvec{d}_0,\rvec{v}_{1}^{k-2}}\\&=& f_{\rvec{v}_{k-1}|\rvec{e}_0,\rvec{d}_0,\rvec{v}_{1}^{k-2}}\label{eq:firstSubErgoPrime}
\end{IEEEeqnarray}
Using (\ref{eq:firstSubErgo}), it can be seen that  $      f_{\rvec{e}_{k},\rvec{v}_{k-1}|\rvec{e}_0,\rvec{d}_0,\rvec{v}_{1}^{k-2}} = f_{\rvec{v}_{k-1}}\int_{\mathbb{R}}f_{\rvec{e}_{k}|\rvec{e}_{k-1},\rvec{e}_0,\rvec{d}_0,\rvec{v}_{1}^{k-1}}f_{\rvec{e}_{k-1}|\rvec{e}_{0},\rvec{d}_0,\rvec{v}_{1}^{k-2}}d\dvec{e}_{k-1}$. Thus, 
\begin{multline}\label{eq:gconv}
        f_{\rvec{e}_{k}|\rvec{e}_0,\rvec{d}_0,\rvec{v}_{1}^{k-1}} =\\ \int_{\mathbb{R}}f_{\rvec{e}_{k}|\rvec{e}_{k-1},\rvec{e}_0,\rvec{d}_0,\rvec{v}_{1}^{k-1}}f_{\rvec{e}_{k-1}|\rvec{e}_{0},\rvec{d}_0,\rvec{v}_{1}^{k-2}}d\dvec{e}_{k-1}.
\end{multline}  As $\rvec{w}_{k-1}\perp (\rvec{d}_{0}^{k-1},\rvec{e}_{0}^{k-1},\rvec{w}_{0}^{t-1},\rvec{x}_{0})$, we have, via (\ref{eq:recursion}) that  $\rvec{e}_{k}$ is conditionally independent of $(\rvec{e}_0,\rvec{d}_{0},\rvec{v}_{0}^{k-2}) $ given $(\rvec{e}_{k-1},\rvec{v}_{k-1})$. Thus $f_{\rvec{e}_{k}|\rvec{e}_{k-1},\rvec{e}_0,\rvec{d}_0,\rvec{v}_{1}^{k-1}}=f_{\rvec{e}_{k}|\rvec{e}_{k-1},\rvec{v}_{k-1}}$. Using (\ref{eq:normalonestep2}) and assuming (\ref{eq:normalStepFormula}) holds for $n=k-1$, the integration in (\ref{eq:gconv}) is given by $f_{\rvec{e}_{k}|\rvec{e}_0,\rvec{d}_0,\rvec{v}_{1}^{k-1}} = \int_{\mathbb{R}} N(\dvec{e}_{k};R\dvec{e}_{k-1}-L\dvec{v}_{k-1},W)N(\dvec{e}_{t-1};\dvec{\mu}_{k-1},\Sigma_{k-1}) d\dvec{e}_{k-1}$. This is the convolution of two multivariate Gaussian PDFs.

Computing the convolution, we have $f_{\rvec{e}_{k}|\rvec{e}_0,\rvec{d}_0,\rvec{v}_{1}^{k-1}} = N(\dvec{e}_{k}; R\mu_{k-1}-L\dvec{v}_{k-1}, R\Sigma_{k-1}R^{\tp}+W)$, e.g. given $(\rvec{e}_0,\rvec{d}_0,\rvec{v}_{1}^{k-1})=(\dvec{e}_0,\dvec{d}_0,\dvec{v}_{1}^{k-1})$,  $\rvec{e}_{k}$ is Gaussian with mean $R\dvec{\mu}_{k-1}-L\dvec{v}_{k-1}$ and covariance $R\Sigma_{k-1}R^{\tp}+W$. As $\mu_{k} = R\mu_{k-1}-L\dvec{v}_{k-1}$ and  $\Sigma_{k}=R\Sigma_{k-1}R^{\tp}+W$ we have the proof. 
\end{proof} The next lemmas describe properties of $\{\Sigma_{n}\}$ and the sequence of functions $\mu_{n}(e_0,d_0,v_{1}^{n-2}):\mathbb{R}^{m}\times (\mathcal{B}^{m}(\Delta))^{n-1}\rightarrow \mathbb{R}^{m}$. We recall a result from system theory. 
\begin{lemma}[Gelfand's Theorem and a Corollary  \cite{dullRobust}]\label{lemm:gelfand} Let $Q\in\mathbb{R}^{m\times m}$ and denote the spectral radius of $Q$ (e.g. the largest of the absolute values of $Q$'s eigenvalues) by $\rho_{\max}(Q)$. Gelfand's Theorem states that $\lim_{n\rightarrow \infty}\left(\lVert Q^{n} \rVert_2\right)^{\frac{1}{n}}  = \rho_{\max}(Q)$. Let $\gamma=(\rho_{\max}(Q)+1)/2$. If $\rho_{\max}(Q)<1$ then $\gamma<1$ and by Gelfand's Theorem $\exists$ $i$ such that $\forall$ $j\ge i$ $\lVert Q^{j} \rVert_2 \le \gamma^{j}$. 
\end{lemma}
\begin{lemma}\label{lemm:sigset}
We have $\Sigma_{n}\succeq W \succ 0_{m\times m}$. Furthermore, there exists a constant $c$ such that for all $n$, $\lVert \Sigma_{n} \rVert_2 \le c$.  
\end{lemma}
\begin{proof}
The first statement is immediate from the assumption that $W\succ 0_{m\times m}$ and the formula for $\Sigma_n$ in Lemma \ref{lemm:normalPDF}. Note that since $\Sigma_{n-1}\preceq\Sigma_{n}$, we have that $\{\lVert\Sigma_{n} \rVert_{2}\}$ is monotonically increasing. Recall that 
$\rho_{\max}(R) < 1$ \cite{RDE_convergence}. Using  Lemma \ref{lemm:gelfand}, it is seen that for some $c<\infty$,  $\lim_{r\rightarrow\infty}\sum\limits_{i=0}^{r}\lVert R^{i} \rVert_{2}^{2} \lVert W \rVert_{2} = c$. The triangle inequality and the submultiplicity of the matrix 2-norm gives $\lVert\Sigma_{n}\rVert_{2} \le  c$.
\end{proof}
\begin{lemma}\label{lemm:muset}
There exist constants $\alpha$ and $\beta$ such that for any $n$ and choice of $\dvec{v}_{1}^{n-2}\in  (\mathcal{B}^{m}(\Delta))^{n-2}$ we have  $\lVert \mu_{n}\left(\dvec{e}_0,\dvec{d}_0,\dvec{v}_{1}^{n-1}\right) \rVert_{2} \le \alpha\lVert M(e_0,d_0) \rVert_{2}+\beta$. 
\end{lemma}
\begin{proof}
The proof is analogous to that of Lemma \ref{lemm:sigset}. Lemma \ref{lemm:gelfand} is used to bound $\lVert R^{n-1} \rVert_2$ and $\sum_{i=0}^{n-2}\lVert R^{i}\rVert_{2}$. We also use the fact that for  $\dvec{v}_{j} \in  \mathcal{B}^{m}(\Delta)$, $\lVert\dvec{v}_{j} \rVert_{2}\le \frac{\Delta\sqrt{m}}{2}$.
\end{proof}  For $\alpha$, $\beta$, and $c$ as defined in Lemmas \ref{lemm:sigset} and \ref{lemm:muset}, define a subset of $\mathbb{R}^{m}\times \mathbb{R}^{m\times m}$ via $
\mathcal{M}(e_0,d_0,m,L,R,\Delta) =\{ (\dvec{\mu},\Sigma)\in\mathbb{R}^{m}\times \mathbb{R}^{m\times m}: \lVert\dvec{\mu}\rVert_2 \le \alpha\lVert M(e_0,d_0) \rVert+\beta, \Sigma\succeq W, \lVert\Sigma\rVert_2\le c  \}$. We have that $\mathcal{M}(e_0,d_0,m,L,R,\Delta)$ is compact (closed and bounded). This helps prove the main result of this section.
\begin{lemma}\label{lemm:exist}
There exists an \textit{invariant} PDF $g_{\infty}$ such that 
\begin{align}\label{eq:invarpdfdef}
    \int_{\mathbb{D}} f_{t+1|t}(\dvec{e}, \dvec{d}|\dvec{x},\dvec{y})g_{\infty}(\dvec{x},\dvec{y})d\dvec{x}d\dvec{y} =  g_{\infty}(\dvec{e},\dvec{d}),
\end{align} and $g_{\infty}(\dvec{x},\dvec{y})>0$ for all $(\dvec{x},\dvec{y})\in\mathbb{D}$.
\end{lemma} 
\begin{proof}
 A set $F\in\mathbb{B}(\mathbb{D})$ is called \textit{weakly transient} with respect to the Markov kernel (\ref{eq:transisionkernel}) if there exists a sequence of positive integers $n_{1}<n_{2}<\dots$ such that 
\begin{align}\label{eq:transientdef}
    \sum_{i=1}^{\infty}\mathbb{P}_{\rvec{e}_{n_{i}},\rvec{d}_{n_i}|\rvec{e}_{0},\rvec{d}_{0}}[F|\rvec{e}_{0}=\dvec{e}_{0},\rvec{d}_{0}=\dvec{d}_{0}]<\infty
\end{align} holds for $\lambda$ almost-every $(\dvec{e}_{0},\dvec{d}_{0})\in\mathbb{D}$. A consequence of \cite[Theorem 5]{ito_invariant} is that there exists an invariant PDF $g_{\infty}$ satisfying (\ref{eq:invarpdfdef}) and $g_{\infty}(a,b)>0$ for all $(a,b)\in\mathbb{D}$ if and only if every weakly transient set $F$ has $\lambda(F) = 0$. To prove the lemma, we proceed via contradiction. We prove that any set $F\in\mathbb{B}(\mathbb{D})$ with $\lambda(F) > 0$ is not weakly transient. Let $F\subset \mathbb{D}$ have $\lambda(F)>0$. Since $\lambda(F)>0$, $F$ must contain an open ball, which in turn must contain a closed rectangle of positive Lebesgue measure, e.g. there exists a $\delta>0$ and point $(\dvec{e}_{F},\dvec{d}_F)\in  F$  such that the set $ H = \{(x,y)\in\mathbb{R}^{m}\times \mathbb{R}^{m} :\lVert x-\dvec{e}_{F}\rVert_{\infty}\le \frac{\delta}{2},\lVert y-d_{F}\rVert_{\infty}\le \frac{\delta}{2} \}$ has $H\subset F$.  Note $\lambda(H)=\delta^{2m}$. Define the ``e-section" of $H$ as the subset $H_{e} = \{x\in\mathbb{R}^{m}:\lVert x-\dvec{e}_{F}\rVert_{\infty}\le \frac{\delta}{2} \}$. As $H\subset F$, the conditional probability satisfies 
\begin{IEEEeqnarray}{rCl}\label{eq:wherewestart}
    \mathbb{P}_{\rvec{e}_{n_{i}},\rvec{d}_{n_i}|\rvec{e}_{0},\rvec{d}_{0}}[F|\dvec{e}_{0},\dvec{d}_{0}] &\ge& \int_{H}f_{t+n_{i}|t}(x,y|\dvec{e}_{0},\dvec{d}_0)dxdy\\ &\ge& \delta^{2m} \inf_{(\dvec{x},\dvec{y})\in H} f_{t+n_{i}|t}(\dvec{x},\dvec{y}|\dvec{e}_{0},\dvec{d}_0), \nonumber\\ &\ge& \frac{\delta^{2m}}{\Delta^{m}}\inf_{\dvec{x}\in H_{e}} f_{\rvec{e}_{n_{i}}|\rvec{e}_{0},\rvec{d}_{0}}(\dvec{x}|\dvec{e}_{0},\dvec{d}_0).\nonumber
\end{IEEEeqnarray} Since the ``minimum is less than the average", we have
\begin{multline}\label{eq:optim_about_to_relax}
 \inf_{\dvec{x}\in H_{e}} f_{\rvec{e}_{n_{i}}|\rvec{e}_{0},\rvec{d}_{0}}(\dvec{x}|\dvec{e}_{0},\dvec{d}_0) \ge \\ \inf_{\substack{\dvec{x}\in H_{e} \\ \dvec{v}_{1}^{n_{i}-1} \in (\mathcal{B}^{m}(\Delta))^{n-1} }} f_{\rvec{e}_{n_{i}}|\rvec{e}_{0},\rvec{d}_{0},\rvec{v}_{1}^{n_{i}-1}}(\dvec{x}|\dvec{e}_{0},\dvec{d}_0,\dvec{v}_{1}^{n_{i}-1}).
\end{multline} 
Using the explicit formula for $f_{\rvec{e}_{n_{i}}|\rvec{e}_{0},\rvec{d}_{0},\rvec{v}_{1}^{n_{i}-1}}$ given by (\ref{eq:normalStepFormula}) in Lemma \ref{lemm:normalPDF}, the right side of (\ref{eq:optim_about_to_relax}) can be written in terms of $\Sigma_{{n}_{i}}$ and $\mu_{n_{i}}(e_0,d_0,v_{1}^{{n}_{i}})$. However, for a fixed $(e_{0},d_{0})$, Lemma \ref{lemm:muset} proves a compact set that contains $\mu_{n_{i}}(e_{0},d_{0},v_{1}^{n_{i}})$ for any choice of $n_{i}$ and $v_{1}^{n_{i}}$. This set is given by  $\mathcal{C}_{\dvec{\mu}} = \{\dvec{\mu} \in \mathbb{R}^{m}: \lVert\mu \rVert_{2} \le  \alpha\lVert M(e_{0},d_{0})\rVert+\beta \}$. Likewise, for any $n_{i}$, $\Sigma_{n_{i}}$ falls within the compact set $\mathcal{C}_{\Sigma} = \{{\Sigma} \in \mathbb{R}^{m\times m}: \Sigma \succeq W, \lVert \Sigma \rVert_{2}\le c  \}$
Thus, we have, for any $n_{i}$
\begin{multline}\label{eq:explicitminimization}
    \inf_{\substack{\dvec{x}\in H_{e} \\ \dvec{v}_{1}^{n_{i}-1} \in (\mathcal{B}^{m}(\Delta))^{n-1} }} f_{\rvec{e}_{n_{i}}|\rvec{e}_{0},\rvec{d}_{0},\rvec{v}_{1}^{n_{i}-1}}(\dvec{x}|\dvec{e}_{0},\dvec{d}_0,\dvec{v}_{1}^{n_{i}-1}) \ge\\  \inf_{\substack{\dvec{x}\in H_{e}\text{, }\dvec{\mu} \in \mathcal{C}_{\dvec{\mu}}\text{, }\Sigma \in \mathcal{C}_{\Sigma}  }}  N(x; \dvec{\mu}, \Sigma). 
\end{multline} Note that the optimization on the right of (\ref{eq:explicitminimization}) does not depend on $n_{i}$. At every point $(x,\mu,\Sigma)$ where $\Sigma \succ 0_{m\times m}$, the positive function $N(x;\mu,\Sigma):  \mathbb{R}^{m}\times \mathbb{R}^{m} \times \mathbb{R}^{m\times m}\rightarrow \mathbb{R}$ is continuous. Note that the subset
 $H_{e}\times \mathcal{C}_{\mu} \times \mathcal{C}_{\Sigma} \subset \mathbb{R}^{m}\times\mathbb{R}^{m} \times  \mathbb{R}^{m\times m}$ is closed and bounded, thus compact, and that $\Sigma \in \mathcal{C}_{\Sigma} \rightarrow \Sigma\succ 0_{m\times m}$. Thus we have that there exists $\epsilon > 0$ such that $ \inf_{\substack{\dvec{x}\in H_{e}, \dvec{\mu} \in \mathcal{C}_{\dvec{\mu}},\text{ } \Sigma \in \mathcal{C}_{\Sigma}  }}  N(x; \dvec{\mu}, \Sigma)\ge \epsilon$. Thus, following the inequalities from (\ref{eq:wherewestart}) through (\ref{eq:explicitminimization}) gives, for any $n_i$ that $     \mathbb{P}_{\rvec{e}_{n_{i}},\rvec{d}_{n_i}|\rvec{e}_{0},\rvec{d}_{0}}[F|\dvec{e}_{0},\dvec{d}_{0}] > \epsilon$. Thus for any $F$ with $\lambda(F)>0$ the terms in the summation (\ref{eq:transientdef}) are bounded from below by $\epsilon$. For any subsequence $\{n_{i}\}$ of $\mathbb{N}$ the summation in (\ref{eq:transientdef}) diverges, which implies $F$ is \textit{not} weakly transient. Thus, if $F$ is weakly transient, it must have $\lambda(F)=0$, which via \cite[Theorem 5]{ito_invariant} proves the lemma. 
\end{proof} The next ensures that the sequence $\{\rvec{e}_{t},\rvec{d}_{t}\}$ converges to $g_{\infty}$.
\begin{lemma}\label{lemm:ergo}
Let $(\rvec{e}_{\infty},\rvec{d}_{\infty})\sim g_{\infty}$, and  $(\rvec{e}_{0},\rvec{d}_{0})$ be drawn from a continuous distribution on $\mathbb{D}$. The sequence of random variables $\{(\rvec{e}_{t},\rvec{d}_{t})\}$ converges weakly to $(\rvec{e}_{\infty},\rvec{d}_{\infty})$.
\end{lemma}
\begin{proof}
We can prove this result using \cite[Theorem 4]{mcmcReview} adapted to the special case when a first-order time-homogeneous Markov chain  admits an invariant PDF (\ref{eq:invarpdfdef}) (in other words. when the chain admits an invariant measure that is absolutely continuous with respect to Lebesgue measure). A Markov chain $\{\rvec{g}_{t}\}$ on a state-space $\mathbb{G}$ is said to be  $\phi$-irreducible if there exists a nonzero $\sigma$-finite measure $\phi$ such that for all measurable $\mathcal{G}\subset \mathbb{G}$ with $\phi(\mathcal{G})>0$ and all initial conditions $\rvec{z}_{0}=\dvec{z}_{0}$ with $z_0\in\mathbb{D}$ we can find an integer $n$ such that $\mathbb{P}_{\rvec{z}_{n}|\rvec{z}_{0}}[\rvec{z}_{n}\in\mathcal{A}|\rvec{z}_{0}= {z}_{0}] > 0$.
 A Markov chain $\{\rvec{g}_{t}\}$ on $\mathbb{G}$ is \textit{aperiodic} if there does not exist $d\ge 1$ and disjoint nonempty measurable subsets $\mathcal{G}_{0},\mathcal{G}_{1}, \dots \mathcal{G}_{d-1}$ such that when $\dvec{g}_{n-1}\in\mathcal{G}_{i}$, then  $ \mathbb{P}_{\rvec{g}_{n}|\rvec{g}_{n-1}}[\rvec{g}_{n}\in\mathcal{G}_{i+1\text{ mod } d }|\rvec{z}_{n-1}=\dvec{z}_{n-1}] = 1$.  A consequence of \cite[Theorem 4]{mcmcReview} is that if the Markov chain $\{\rvec{g}_{t}\}$ admits an invariant PDF $g_{\infty}$, is $\phi$-irreducible, aperiodic, and $\rvec{g}_{0}$ is a continuous random variable then the $\rvec{g}_{i}$ converge weakly to the measure induced by the invariant PDF. Consider the Markov chain $\{\rvec{e}_{t},\rvec{d}_{t}\}$ on $\mathbb{D}$, and recall that the Lebesgue measure $\lambda$ on $\mathbb{D}$ is countably generated. Take $\mathcal{A}\subset \mathbb{D}$ with $\lambda(\mathcal{A})>0$, and let $(\dvec{e}_{0},\dvec{d}_{0})\in \mathbb{D}$ be arbitrary. We have 
\begin{multline}\label{eq:irreducibilitypf}
    \mathbb{P}_{\rvec{e}_{1},\rvec{d}_{1}|\rvec{e}_{0},\rvec{d}_{0}}[(\rvec{e}_{1},\rvec{d}_{1})\in\mathcal{A}|(\rvec{e}_{0},\rvec{d}_{0})= (\dvec{e}_{0},\dvec{d}_{0})]   =\\ \int_{\mathcal{A}} f_{t+1|t}(\dvec{e},\dvec{d}|\dvec{e}_{0},\dvec{d}_{0})ddde.
\end{multline} It is immediate from the defition of $f_{t+1|t}$ in (\ref{eq:transisionkernel}) and the fact that $\mathcal{A}\subset \mathbb{D}$ that $  \mathbb{P}_{\rvec{e}_{1},\rvec{d}_{1}|\rvec{e}_{0},\rvec{d}_{0}}[(\rvec{e}_{1},\rvec{d}_{1})\in\mathcal{A}|(\rvec{e}_{0},\rvec{d}_{0})  (\dvec{e}_{0},\dvec{d}_{0})]\ge 0$. Thus the Markov chain $\{\rvec{e}_{t},\rvec{d}_{t}\}$ is $\lambda-$irreducible. 

Via the same logic, it can be shown by contradiction that the chain is aperiodic. Assume that the chain is periodic, e.g. that for $d\ge2$, there exist disjoint nonempty measurable subsets $\mathcal{A}_{0},\mathcal{A}_1,\dots, \mathcal{A}_{d-1}\subset \mathbb{D}$ satisfying such that if $(\dvec{e}_{n-1},\dvec{d}_{n-1})\in\mathcal{A}_{i}$ then   $\mathbb{P}_{\rvec{e}_{n},\rvec{d}_{n}|\rvec{e}_{n-1},\rvec{d}_{n-1}}[(\rvec{e}_{n},\rvec{d}_{n})\in\mathcal{A}_{i+1\text{ mod } d }|\rvec{z}_{n-1}=\dvec{z}_{n-1}] = 1$.  Assume that $\mathcal{A}_{i}$ has $\lambda(\mathcal{A}_{i})>0$ (since the sets $\{\mathcal{A}_k\}$ partition $\mathbb{D}$, at least one of the sets must have positive Lebesgue measure). Take $(e_t,d_t)\in\mathcal{A}_{i}$, 
and let $j = i+1 \mod d$. By the assumption of periodicity
\begin{multline}\label{eq:periodassum}
    \mathbb{P}_{t+1|t}[(\rvec{e}_{t+1},\rvec{d}_{t+1})\in\mathcal{A}_{j}|(\rvec{e}_t,\rvec{d}_{t})=(\dvec{e}_t,\dvec{d}_t)] = 1. 
\end{multline} If $\mathcal{A}_{i}$ has $\lambda(\mathcal{A}_{i})>0$, for any $(e_t,d_t)$ our proof of irreducibility guarantees that $\exists$ $\epsilon>0$ such that $\mathbb{P}[(\rvec{e}_{t+1},\rvec{d}_{t+1})\in \mathcal{A}_{i}|\rvec{e}_t=\dvec{e}_t, \rvec{d}_t=\dvec{d}_t] \ge\epsilon$. As we assumed $\mathcal{A}_{i}\cap\mathcal{A}_{j}=\emptyset$, this contradicts (\ref{eq:periodassum}). Thus the Markov chain $\{\rvec{e}_{i},\rvec{d}_{i}\}$ is aperiodic, and by \cite[Theorem 4]{mcmcReview} we have the proof. 
\end{proof}
The next lemma characterizes the invariant PDF $g_{\infty}$.
\begin{corollary}\label{corr:invarprops}
Assume $(\rvec{e},\rvec{d})\sim g_{\infty}$. Denote the marginal PDF of $\rvec{e}$ via $g_{\rvec{e}_{\infty}}(e) = \int_{\mathcal{B}^m_{\Delta}}g_{\infty}(e,\delta)d\delta$.
We have that $\rvec{e}\perp \rvec{d}$ and $\rvec{d}\sim \text{Uniform}[-\Delta/2,\Delta/2]$. This implies that, $g_{\infty}\colon\mathbb{D}\rightarrow \mathbb{R}$ factorizes via $g_{\infty}(e,d) = g_{\rvec{e}_{\infty}}(e)/\Delta^m$ for $(e,d)\in\mathbb{D}$.
\end{corollary}
\begin{proof}
Let $\mathcal{A}_1$ be an open ball in $\mathbb{R}^{m}$ and $\mathcal{A}_{2}$ be an open ball inside $\mathcal{B}^m(\Delta)$. Open intervals like  $\mathcal{A}_{1}\times\mathcal{A}_{2}$ form a $\pi$ system that generates the $\sigma-$algebra $\mathbb{B}(\mathbb{D})$. Using the definition of the transition kernel (\ref{eq:transisionkernel}) and the invariant PDF (\ref{eq:invarpdfdef}), it can be shown that if $\mathcal{P}=\mathcal{A}_{1}\times \mathcal{A}_{2}$ then, $ \mathbb{P}[(\rvec{e},\rvec{d})\in \mathcal{P}] = \int_{\mathcal{A_{1}}}g_{\rvec{e}_{\infty}}(e)de\frac{\lambda(\mathcal{A}_{2})}{\Delta^m}$. By Dynkin's $\pi-\lambda$ theorem, this proves that $\rvec{e}\perp \rvec{d}$ (see e.g., \cite[Prop. 2.13]{zit}).  
\end{proof} Let $(\rvec{e},\rvec{d})\sim g_{\infty}$, and define $\rvec{q} = Q_{\Delta}(C\rvec{e}+\rvec{d})$. We now prove that $\rvec{q}_{t}$ converges to $\rvec{q}$ in the KL sense. This is not obvious a priori as $\rvec{q}_{t}$ and $\rvec{q}$ have countably infinite support. 
\begin{lemma}\label{lemm:klconv}
$D_{\mathrm{KL}}(\rvec{q}_{t}||\rvec{q}) <\infty$  and $\underset{t\rightarrow\infty}{\lim} D_{\mathrm{KL}}(\rvec{q}_{t}||\rvec{q}) =0$.
\end{lemma}
\begin{proof}
Using the properties of the invariant measure and a data processing inequality for f-divergences \cite[Theorem 2.2 (6)]{polyanbk} one can prove that the sequence of relative entropies is monotonically decreasing, i.e., that $D_{\mathrm{KL}}(\rvec{q}_{t+1}||\rvec{q})\le D_{\mathrm{KL}}(\rvec{q}_{t}||\rvec{q})$ for all $t$. Since $D_{\mathrm{KL}}(\rvec{q}_{t}||\rvec{q}) \ge 0$, the limit as $t\rightarrow\infty$ exists. From applying the data processing inequality for KL divergences several times (cf. \cite[Theorem 2.2 (6) ]{polyanbk}), and the fact that both $\rvec{d}_{t}\perp \rvec{e}_{t}$ and by Corollary \ref{corr:invarprops} that both $\rvec{d}\perp\rvec{e}$ and $\rvec{d}_{t}$ and $\rvec{d}$ are identically distributed, we have that $D_{\mathrm{KL}}(\rvec{q}_{t}||\rvec{q})\le D_{\mathrm{KL}}(\rvec{e}_{t}||\rvec{e})$.
 Details can be found in the proof of \cite[Lemma IV.10]{ourJSAIT}. The rest of the proof follows from bounding $D_{\mathrm{KL}}(\rvec{e}_{t}||\rvec{e})$.

Let $\{\rvec{\nu}_{t}\}$ denote an IID sequence of random variables uniformly distributed on $\mathcal{B}^{m}(\Delta)$. Likewise, let $\{\rvec{\omega}_{t}\}$ be IID with $\rvec{\omega}_{t}\sim\mathcal{N}(0,W)$, and let $\rvec{\lambda}\sim\mathcal{N}(0,X_{0})$. Assume $\{\rvec{\omega}_{t}\}$, $\{\rvec{\nu}_{t}\}$, and $\rvec{\lambda}$ are mutually independent.  Let  ``$\rvec{a}\overset{D}{=}\rvec{b}$" denote ``$\rvec{a}$ and $\rvec{b}$ are identically distributed". By unwrapping the recursive definition of $\{\rvec{e}_t\}$ in (\ref{eq:recursion}), we have
    $\rvec{e}_{t} \overset{D}{=} R^{t}\rvec{\lambda}+\sum_{i=0}^{t-1}R^{i}(\rvec{\omega}_{i}-L\rvec{\nu}_{i})$. 
 Likewise, by definition of $\rvec{e}$, we have  that both $\rvec{e} \overset{D}{=} \lim_{t\rightarrow\infty}R^{t}\rvec{\lambda}+\sum_{i=0}^{t-1}R^{i}(\rvec{\omega}_{i}-L\rvec{\nu}_{i})$ and  $\rvec{e} \overset{D}{=}\lim_{t\rightarrow\infty}\sum_{i=0}^{t-1}R^{i}(\rvec{\omega}_{i}-L\rvec{\nu}_{i})$
following from the weak convergence guaranteed by Lemma \ref{lemm:ergo}. 
Define the random variables $\rvec{g}_{\le t} =  \sum_{i=0}^{t-1}R^{i}\rvec{\omega}_{i}$, $    \rvec{u}_{\le t} =  -\sum_{i=0}^{t-1}R^{i}L\rvec{\nu}_{i}$, and $\rvec{s}_{>t} =  \lim_{r\rightarrow \infty}\sum_{i=t}^{r}R^{i}(\rvec{\omega}_{i}-L\rvec{\nu}_{i})$. Note that the limit in the definition of  $\rvec{s}_{>t}$ is well defined via Lemma \ref{lemm:gelfand} in concert with Kolmogorov's two-series theorem. By definition, $\rvec{e}_{t} \overset{D}{=} R^{t}\rvec{\lambda}+ \rvec{g}_{\le t} +\rvec{u}_{\le t}$ and $\rvec{e} \overset{D}{=}  \rvec{g}_{\le t}+ \rvec{u}_{\le t} + \rvec{s}_{> t}$.
 Note that $\rvec{g}_{\le t}\sim \mathcal{N}(0,\sum\limits_{i=0}^{t-1}R W R^{\tp})$. We have
\begin{IEEEeqnarray}{rCl}
    D_{\mathrm{KL}}(\rvec{e}_{t}|| \rvec{e} ) &\le& D_{\mathrm{KL}}(R^{t}\rvec{\lambda}+\rvec{g}_{\le t}||\rvec{g}_{\le t} +\rvec{s}_{> t} )\label{eq:usedidiv} \\ &\le& D_{\mathrm{KL}}(R^{t}\rvec{\lambda}+\rvec{g}_{\le t} || \rvec{g}_{\le t}+\rvec{s}_{> t} | \rvec{s}_{> t} ),\label{eq:condincdiv}  
\end{IEEEeqnarray} where (\ref{eq:usedidiv}) follows from the data processing inequality for f-divergences and (\ref{eq:condincdiv}) follows since conditioning increases KL divergence (see \cite[Theorem 2.2 (5)]{polyanbk}). Given $\rvec{s}_{> t}=s$,  (\ref{eq:condincdiv}) simplifies to a KL divergence between multivariate Gaussians. Let $S_t =\sum_{i=0}^{t-1}R W R^{\tp}$ and $\overline{S}_{t}=S_{t}+R^{t}X_{0}(R^{\tp})^{t}$. Since $\rvec{\lambda}\perp\rvec{g}_{\le t}$ by construction, $R^{t}\rvec{\lambda}+\rvec{g}_{\le t}\sim \mathcal{N}(0,\overline{S}_{t})$. Recall also that $\rvec{g}_{\le t}\perp\rvec{s}_{> t}$  by construction. We have, then that $D_{\mathrm{KL}}(R^{t}\rvec{\lambda}+\rvec{g}_{\le t} || \rvec{g}_{\le t}+\rvec{s}_{> t} | \rvec{s}_{> t} = s ) = D_{\mathrm{KL}}(\mathcal{N}(0,\overline{S}_{t})|| \mathcal{N}(s,S_{t}))$.  Let $f_{t}:\mathbb{R}^{m}\rightarrow \mathbb{R}$ be defined via
\begin{align}\label{eq:ffuncdef}
    f_t(s) = \text{Tr}(S_{t}^{-1}\overline{S}_{t})+\dvec{s}^{\tp}S_{t}^{-1}\dvec{s}+
   \ln{\left(\det{S_{t}}/\det{\overline{S}_{t}}\right)}.
\end{align} We have that $ D_{\mathrm{KL}}(\mathcal{N}(0,\overline{S}_{t})|| \mathcal{N}(s,S_{t})) = \frac{1}{2}(f_t(s)-m)$, where the divergence is expressed in nats. Thus, via (\ref{eq:condincdiv}),
\begin{IEEEeqnarray}{rCl}
    D_{\mathrm{KL}}(\rvec{q}_{t} || \rvec{q} ) &\le& \frac{\mathbb{E}\left[(f_{t}(\rvec{s}_{>t})-m) \right]}{2}.\label{eq:expectedkl}
\end{IEEEeqnarray} Since $\rvec{s}_{>t}\in \mathcal{L}^2$ and $S_{t},\overline{S}_{t}\succ W$, it is clear that the right-hand side of (\ref{eq:expectedkl}) is finite for all $t$. We now demonstrate that $\lim_{t\rightarrow\infty}\mathbb{E}\left[(f_{t}(\rvec{s}_{>t})-m) \right]=0$, attacking the terms in (\ref{eq:ffuncdef}) one-at-a-time. As $\overline{S}_{t}\succeq {S}_{t} \succeq W \succ 0_{m\times m}$, the determinants $\det{{S}_{t}}$ and $\det{\overline{S}_{t}}$ are bounded strictly away from $0$. Since $\rho_{\max}(R)< 1$, we have that 
$Q=\lim_{t\rightarrow\infty }S_{t}$ is well defined, and that $\lim_{t\rightarrow\infty }\overline{S}_{t}=Q$. Thus
\begin{align}\nonumber
    \lim_{t\rightarrow\infty} \ln\left(\det{S_{t}}/\det{\overline{S}_{t}}\right) =0\text{ and }\lim_{t\rightarrow\infty} \text{Tr}(S_{t}^{-1}\overline{S}_{t})=m.
\end{align}
 We now prove $\lim_{t\rightarrow \infty}\mathbb{E}[\rvec{s}^{\tp}S_{t}^{-1}\rvec{s}]=0$. Fix $t$ and let $\rvec{p}_{t;r} = \sum_{i=t}^{r}R^{i}(\rvec{\omega}_{i}-L\rvec{\nu}_{i})$. For any $t$, $ \lim_{r\rightarrow \infty}\rvec{p}_{t;r}\rvec{p}_{t;r}^{\tp} = \rvec{s}_{>t}\rvec{s}_{>t}^{\tp}$. These limits are well defined by Kolmogorov's two-series theorm. Let $U = W+L\frac{\Delta^2}{12}L^{\tp}$. We have
\begin{IEEEeqnarray}{rCl}
    \mathbb{E}[\rvec{s}^{\tp}_{>t} \overline{S}^{-1}_{t}\rvec{s}_{>t} ]&=& \mathbb{E}[\text{Tr}\left(\overline{S}^{-1}_{t}\rvec{s}_{>t}\rvec{s}^{\tp}_{>t}\right)] \\ &\le&\text{Tr}\left(\overline{S}^{-1}_{t}\underset{r\rightarrow\infty}{\lim\inf}\text{ }\mathbb{E}\left[\rvec{p}_{t;r}\rvec{p}_{t;r}^{\tp} \right]\right)\label{eq:applytraceandfatou} .
\end{IEEEeqnarray} where (\ref{eq:applytraceandfatou}) follows from Fatou's lemma, and linearity. We have $\mathbb{E}\left[\rvec{p}_{t;r}\rvec{p}_{t;r}^{\tp} \right] = \sum_{i=t}^{r}R^{i}U(R^{\tp})^{i}$,  thus
\begin{multline}\label{eq:klapplygelf}
        \lim_{r\rightarrow\infty} \mathbb{E}\left[\rvec{p}_{t;r}\rvec{p}_{t;r}^{\tp} \right] =  R^{t}\left(\lim_{r\rightarrow\infty}\sum_{i=0}^{r}R^{i}U(R^{\tp})^{i}\right) (R^{\tp})^t,\nonumber
\end{multline} where the limit exists via Lemma \ref{lemm:gelfand}. Let $N = \left(\lim_{r\rightarrow\infty}\sum_{i=0}^{r}R^{i}\left(W+L\frac{\Delta^2}{12}L^{\tp} \right)(R^{\tp})^{i}\right)$. We have established that $\mathbb{E}[\rvec{s}^{\tp}_{>t} \overline{S}^{-1}_{t}\rvec{s}_{>t} ] \le \text{Tr}(\overline{S}^{-1}_{t}R^{t}N(R^{\tp})^{t})$. Thus since  $\rho_{\max}(R)<1$ taking the limit as $t\rightarrow \infty$ proves  $\underset{t\rightarrow\infty}{\lim}\mathbb{E}[\rvec{s}^{\tp}_{>t} \overline{S}^{-1}_{t}\rvec{s}_{>t} ]=0$. Thus, in conclusion $\lim_{t\rightarrow \infty } D(\rvec{q}_{t}||\rvec{q}) \le \lim_{t\rightarrow \infty } \frac{\mathbb{E}\left[(f_{t}(\rvec{s}_{>t})-m) \right]}{2} \le 0$. 
\end{proof} 
\vspace{-.5cm}
\section{Conclusion} 
For $m$ dimensional plants, we have proven the existence of a time-invariant data compression architecture and controller that can achieves any feasible LQG control performance given a feedback bitrate within approximately $2+1.26m$ bits of a known lower bound. Using the ``Shannon-type" source codes described in \cite[Section IV.A]{ourJSAIT}, it can be shown that this overhead can be reduced to $1+1.26m$ bits. The difference between the lower bound and the rate achieved by the dither free, but time-varying, achievability architecture in \cite{kostinaTradeoff} is logarithmic in the plant dimension (i.e. $\mathcal{O}(\log(m))$). This follows from the use of more sophisticated lattice quantizers \cite{kostinaTradeoff}. An opportunity for future work is to demonstrate TI achievability without the use of dithering and with more sophisticated quantization. 

While we have proved that the minimum bitrate can nearly be achieved with a time-invariant source codec, there still remains the problem of developing a practical implementation of this source coding scheme. Numerical experiments would lend credence to these theoretical results.  While the time-varying achievability approaches in  \cite{tanakaISIT} \cite{kostinaTradeoff} require the precise construction of codebooks in accordance with the PMF of the quantizer output, the long-term analyses of Section \ref{sec:ergopfs} could be useful in proving bounds on the performance of adaptive compression schemes based on, for example, \cite{frenchExp}. 

\vspace{-.09cm}
\bibliographystyle{IEEEtran}
\bibliography{refs.bib}
\end{document}